\documentclass{article}
\usepackage[margin=1in]{geometry}
\usepackage{authblk}
\RequirePackage{amsthm,amsmath,amsfonts,amssymb}

\RequirePackage[numbers]{natbib}
\RequirePackage[colorlinks,citecolor=blue,urlcolor=blue]{hyperref}
\RequirePackage{graphicx}
\usepackage{subcaption}
\usepackage{tikz}
\usepackage{comment}
\usepackage{amsmath}  
\usepackage{amsfonts} 
\usepackage{multirow} 
\usepackage{array}    

\numberwithin{equation}{section}

\newtheorem{theorem}{Theorem}

\newtheorem{proposition}[theorem]{Proposition}
\newtheorem{corollary}{Corollary}

\theoremstyle{remark}
\newtheorem{remark}{Remark}

\title{ZAGREB CONNECTION INDICES ON POLYOMINO CHAINS AND RANDOM POLYOMINO CHAINS}

\author[1]{Saylé Sigarreta}
\author[1]{Hugo Cruz-Suárez}

\affil[1]{BUAP, Departamento de Probabilidad y Estadística,
México.}

\begin{document}

\maketitle

\begin{abstract}
In this manuscript, we delve into the exploration of the first and second Zagreb connection indices of both polyomino chains and random polyomino chains. Our methodology relies on the utilization of Markov chain theory. Within this framework, the article thoroughly examines precise formulas and investigates extreme values. Leveraging the derived formulas, we further explore and elucidate the long-term behavior exhibited by random polyomino chains.
\end{abstract}

\textbf{Keywords}: mathematical chemistry, topological indices, Markov chain, random graphs, polyomino chains

\textbf{MSC classes}: 05C50, 05C80

\maketitle
\section{Introduction}

\noindent

All graphs considered in this paper are finite, simple, and connected unless otherwise specified. A graph, denoted as $G=(V,E)$, consists of sets where the elements of $V$ represent the nodes of the graph, and the elements of $E$ represent its edges. Undoubtedly, concepts from graph theory hold vast potential for various applications. For example, a chemical graph serves as a model for representing chemical systems: atoms are represented as nodes of the graph, while molecular bonds are depicted as edges. In this context, topological indices serve to quantify the structural information inherent in the graph, regardless of the specific numbering assigned to nodes and edges. In 1947, Wiener pioneered the use of topological indices to investigate certain physico-chemical properties of alkanes, thereby laying the foundation for chemical graph theory  \cite{bd20}.  Since then, numerous topological indices have been introduced and extensively studied to enhance our comprehension of molecular structure \cite{i1,bd21,v20}. Nowadays, the characterization of molecular structure through topological indices is a fundamental aspect of chemical graph theory. Consequently, this theory plays a pivotal role in designing molecular structures tailored to exhibit desired physico-chemical or biological properties. 

Particularly, the first and second Zagreb indices of a graph were introduced by chemists Gutman, Trinajsti\'c and Rusci\'c \cite{bd3,bd8}. Subsequently, indicating their significance, numerous researchers have further explored the concepts initially presented in these seminal papers. For instance, in a recent publication \cite{bd9}, the authors introduced graph invariants known as the first and second leap Zagreb indices, also referred to as Zagreb connection indices. These indices are derived from the concept of the second degree of a vertex $v$, defined as follows:

$$ Z C_1(G) =\sum_{v \in V(G)} \tau_{v}^2,$$
$$Z C_2(G) =\sum_{uv \in E(G)} \tau_{u} \tau_{v},$$
here, $\tau_v$ denotes the count of a vertex's second neighbors, encompassing all vertices in $G$ situated at a distance of two from $v$. Additionally, the authors investigated several general mathematical properties associated with these indices and their connection with other pertinent quantities within chemical graph theory was elucidated in \cite{bd10}. 

On the other hand, in the realm of random chain analysis, topological indices have emerged as a dynamic research area over the past two decades \cite{bd22}. Numerous topological indices have been scrutinized across various types of random chains, including random cyclooctane chains \cite{bd23,bd24}, random polyphenyl chains \cite{bd25,bd26}, random phenylene chains \cite{bd27,bd28}, random spiro chains \cite{bd29,bd30,bd41}, random hexagonal chains \cite{bd31,bd32,bd13}, and random polyomino chains \cite{bd14,bd35}. Furthermore, it is noteworthy that the Zagreb connection indices have been investigated within the context of random cyclooctatetraene chains, random polyphenyl chains, and triangular chain structures \cite{bd11, bd12}, representing an intersection of these aforementioned research endeavors.

\noindent

Specifically, polyomino systems, characterized by configurations of squares connected edge-to-edge, hold significant historical importance \cite{bd37,bd38} and have found widespread applications across various scientific disciplines. In Chemistry, their utility lies in their capacity to represent and analyze intricate chemical structures, including polymers, crystal structures, and specific organic molecules \cite{bd35,bd34,bd30,v27}. Notably, within this system, the polyomino chain stands out, where adjacent regular cells form a path by connecting their centers.

Within this framework, a square is classified as "terminal" if it has only one adjacent square, "medial" if it has two adjacent squares without any vertex of degree 2, and "kink" if it has two adjacent squares with a vertex of degree 2. A segment $s$ is defined as a maximum linear chain along with a kink and/or terminal square. The count of squares within a segment is denoted by $l(S)$ and indicates the length of the segment. By the way, a linear chain $Li_{n}$ comprises precisely one segment, while a zig-zag chain $Z_n$ consists of polyomino chains where all segments have lengths of 2 (see Figure \ref{f0}).

 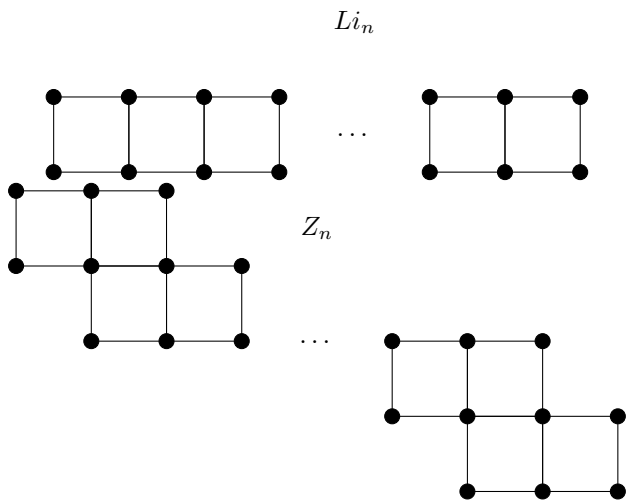
\begin{figure}
\centering

\begin{tikzpicture}


\draw (0,0) rectangle (1,1);

\draw (1,0) rectangle (2,1);

\draw (2,0) rectangle (3,1);

\node[fill=black, draw, circle, inner sep=2pt] (n1) at (0,0) {};
\node[fill=black, draw, circle, inner sep=2pt] (n2) at (1,0) {};
\node[fill=black, draw, circle, inner sep=2pt] (n3) at (0,1) {};
\node[fill=black, draw, circle, inner sep=2pt] (n4) at (1,1) {};
\node[fill=black, draw, circle, inner sep=2pt] (n5) at (2,0) {};
\node[fill=black, draw, circle, inner sep=2pt] (n6) at (3,0) {};
\node[fill=black, draw, circle, inner sep=2pt] (n7) at (2,1) {};
\node[fill=black, draw, circle, inner sep=2pt] (n8) at (3,1) {};

\node at (4,0.5) {$\dots$};
\node at (4,2) {$Li_n$};
\draw (5,0) rectangle (6,1);
\draw (6,0) rectangle (7,1);

\node[fill=black, draw, circle, inner sep=2pt] (n9) at (5,0) {};
\node[fill=black, draw, circle, inner sep=2pt] (n10) at (6,0) {};
\node[fill=black, draw, circle, inner sep=2pt] (n11) at (7,0) {};
\node[fill=black, draw, circle, inner sep=2pt] (n12) at (5,1) {};
\node[fill=black, draw, circle, inner sep=2pt] (n13) at (6,1) {};
\node[fill=black, draw, circle, inner sep=2pt] (n14) at (7,1) {};

\end{tikzpicture}
\hspace{1cm} 
\begin{tikzpicture}

\draw (0,0) rectangle (1,1);

\draw (1,0) rectangle (2,1);

\node[fill=black, draw, circle, inner sep=2pt] (n1) at (0,0) {};
\node[fill=black, draw, circle, inner sep=2pt] (n2) at (1,0) {};
\node[fill=black, draw, circle, inner sep=2pt] (n3) at (0,1) {};
\node[fill=black, draw, circle, inner sep=2pt] (n4) at (1,1) {};
\node[fill=black, draw, circle, inner sep=2pt] (n5) at (2,0) {};
\node[fill=black, draw, circle, inner sep=2pt] (n6) at (2,1) {};

\draw (1,-1) rectangle (2,0);
\draw (2,-1) rectangle (3,0);

\node[fill=black, draw, circle, inner sep=2pt] (n7) at (1,-1) {};
\node[fill=black, draw, circle, inner sep=2pt] (n8) at (2,-1) {};
\node[fill=black, draw, circle, inner sep=2pt] (n9) at (3,-1) {};
\node[fill=black, draw, circle, inner sep=2pt] (n10) at (3,0) {};

\node at (4,-1) {$\dots$};
\node at (4,0.5) {$Z_n$};

\draw (5,-2) rectangle (6,-1);
\draw (6,-2) rectangle (7,-1);
\draw (6,-3) rectangle (7,-2);
\draw (7,-3) rectangle (8,-2);

\node[fill=black, draw, circle, inner sep=2pt] (n11) at (5,-2) {};
\node[fill=black, draw, circle, inner sep=2pt] (n12) at (6,-1) {};
\node[fill=black, draw, circle, inner sep=2pt] (n13) at (6,-2) {};
\node[fill=black, draw, circle, inner sep=2pt] (n14) at (7,-1) {};
\node[fill=black, draw, circle, inner sep=2pt] (n15) at (6,-3) {};
\node[fill=black, draw, circle, inner sep=2pt] (n16) at (7,-2) {};
\node[fill=black, draw, circle, inner sep=2pt] (n17) at (7,-3) {};
\node[fill=black, draw, circle, inner sep=2pt] (n18) at (8,-2) {};
\node[fill=black, draw, circle, inner sep=2pt] (n19) at (5,-1) {};
\node[fill=black, draw, circle, inner sep=2pt] (n20) at (8,-3) {};

\end{tikzpicture}

 \caption{The linear chain and the zigzag chain.}
    \label{f0}

\end{figure}

\noindent

Transitioning to the random context, a random polyomino chain ($RPC_{n}=RPC(n,p)$) can be generated as follows: Figures \ref{f1} illustrate $RPC_{n}$ for $n=1$ and $n=2$. When $n \geq 3$, a new square can be added in two ways, leading to $RPC_{n}^{1}$ and $RPC_{n}^{2}$ with probabilities $p$ and $1-p$, respectively (see Figure \ref{f2}), where $0< p < 1$. Hence, a topological index of a random polyomino chain at time $n$ becomes a random variable.

Continuing along the same vein, it is important to highlight that, the calculation of different topological indices, extreme problems, and exploring topics, such as perfect pairings, dimer covering and their connection with caterpillar trees, among others \cite{bd15,bd16,bd17,bd18,bd19, v14,v15,v16,v17,bd39} are active investigations on chains of polyominoes and chains of random polyominoes. 

Inspired by the aforementioned research and our previous investigation of certain topological indices based on degree over random polyomino chains and polyomino chains \cite{bd14}, the primary objective of this manuscript is to establish explicit formulas for the first and second Zagreb connection indices. Additionally, we aim to analyze their extreme values within the context of polyomino chains. Under the random framework, our goal is to determine the expected values, variances, and asymptotic behavior.

\begin{figure}[h!]

 \centering
\begin{tikzpicture}


\draw (0,0) rectangle (1,1);

\draw (3,0) rectangle (4,1);

\draw (4,0) rectangle (5,1);

\node[fill=black, draw, circle, inner sep=2pt] (n1) at (0,0) {};
\node[fill=black, draw, circle, inner sep=2pt] (n2) at (1,0) {};
\node[fill=black, draw, circle, inner sep=2pt] (n3) at (0,1) {};
\node[fill=black, draw, circle, inner sep=2pt] (n4) at (1,1) {};
\node[fill=black, draw, circle, inner sep=2pt] (n5) at (3,0) {};
\node[fill=black, draw, circle, inner sep=2pt] (n6) at (3,1) {};
\node[fill=black, draw, circle, inner sep=2pt] (n7) at (4,0) {};
\node[fill=black, draw, circle, inner sep=2pt] (n8) at (4,1) {};
\node[fill=black, draw, circle, inner sep=2pt] (n7) at (5,1) {};
\node[fill=black, draw, circle, inner sep=2pt] (n8) at (5,0) {};

\end{tikzpicture}

    \caption{The graphs of $RPC_{1}$ and $RPC_{2}$.}

    \label{f1}
\end{figure}
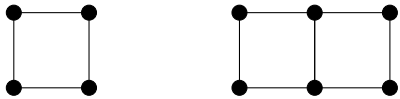

\begin{figure}[h!]
 \centering
\begin{tikzpicture}
  \clip (-3,-2) rectangle (9,3);

\draw (0,0) rectangle (1,1);

\draw (1,0) rectangle (2,1);

\draw (2,0) rectangle (3,1);

\node[fill=black, draw, circle, inner sep=2pt] (n1) at (0,0) {};
\node[fill=black, draw, circle, inner sep=2pt] (n2) at (1,0) {};
\node[fill=black, draw, circle, inner sep=2pt] (n3) at (0,1) {};
\node[fill=black, draw, circle, inner sep=2pt] (n4) at (1,1) {};
\node[fill=black, draw, circle, inner sep=2pt] (n5) at (2,0) {};
\node[fill=black, draw, circle, inner sep=2pt] (n6) at (3,0) {};
\node[fill=black, draw, circle, inner sep=2pt] (n7) at (2,1) {};
\node[fill=black, draw, circle, inner sep=2pt] (n8) at (3,1) {};
  \draw[ dashed] (n1) to  [out=180, in=90,looseness=5] (n4);

\draw (5,0) rectangle (6,1);
\draw (6,0) rectangle (7,1);
\draw (6,-1) rectangle (7,0);

\node[fill=black, draw, circle, inner sep=2pt] (n9) at (5,0) {};
\node[fill=black, draw, circle, inner sep=2pt] (n10) at (6,0) {};
\node[fill=black, draw, circle, inner sep=2pt] (n11) at (7,0) {};
\node[fill=black, draw, circle, inner sep=2pt] (n12) at (5,1) {};
\node[fill=black, draw, circle, inner sep=2pt] (n13) at (6,1) {};
\node[fill=black, draw, circle, inner sep=2pt] (n14) at (7,1) {};
\node[fill=black, draw, circle, inner sep=2pt] (n15) at (6,-1) {};
\node[fill=black, draw, circle, inner sep=2pt] (n16) at (7,-1) {};
\draw[ dashed] (n9) to  [out=180, in=90,looseness=5] (n13);

\node at (2.2,0.5) {\tiny $p$};
\node at (6.5,-0.2) {\tiny $1-p$};
\node at (1.5,-0.5) {\small $RPC_{n}^{1}$};
\node at (6.5,-1.5) {\small $RPC_{n}^{2}$};
\end{tikzpicture}
   \caption{The two link ways for $RPC_{n} (n \geq 3)$.}
    \label{f2}
\end{figure}
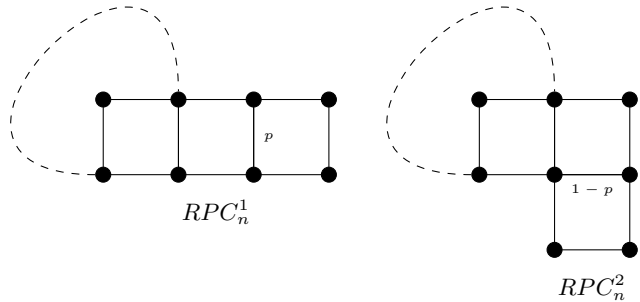

\section{Polyomino Chain}

\noindent
 This section aims to derive explicit mathematical expressions for computing the first Zagreb connection index and the second Zagreb connection index for a polyomino chain with $n$ squares ($PC_{n}$) denoted as $ZC1_{n}$ and $ZC2_n$, respectively. Let $m \geq 1$ and $i \in\{1,2, \ldots, m\}$ be given. It is worth noting that a polyomino chain $PC_{n}$ consists of a sequence of segments $s_{1}, s_{2}, \ldots, s_{m}$ (refer to Figure \ref{f5}) with lengths $l\left(s_{i}\right)=l_{i}$, satisfying the condition that $\sum_{i=1}^{m}l_{i}=n+m-1$. Moreover, a chain of polyominoes $PC_{n}$ induces a sequence of 1s and 2s as follows: If we denote $L_{n}$ as a random variable with rank $\{1,2\}$ such that $p= \mathbb{P}(L_{n} = 1)$, i.e., $L_{n}$ represents the selected link at time $n$, then we have a link of type $j$ at time $n$ if $L_n=j$. Having established this information, we proceed to present the two fundamental results of this section.

 \begin{figure}[h!]
    \centering
   \begin{tikzpicture}

\draw (0,0) rectangle (1,1);

\draw (1,0) rectangle (2,1);
\draw (2,0) rectangle (3,1);
\draw (3,0) rectangle (4,1);
\draw (3,-1) rectangle (4,0);
\draw (3,-2) rectangle (4,-1);
\draw (4,-2) rectangle (5,-1);
\draw (4,-3) rectangle (5,-2);
\draw (5,-3) rectangle (6,-2);
\draw (6,-3) rectangle (7,-2);

\node[fill=black, draw, circle, inner sep=2pt] (n1) at (0,0) {};
\node[fill=black, draw, circle, inner sep=2pt] (n2) at (1,0) {};
\node[fill=black, draw, circle, inner sep=2pt] (n3) at (0,1) {};
\node[fill=black, draw, circle, inner sep=2pt] (n4) at (1,1) {};
\node[fill=black, draw, circle, inner sep=2pt] (n5) at (2,0) {};

\node[fill=black, draw, circle, inner sep=2pt] (n7) at (2,1) {};
\node[fill=black, draw, circle, inner sep=2pt] (n7) at (3,0) {};
\node[fill=black, draw, circle, inner sep=2pt] (n7) at (3,1) {};
\node[fill=black, draw, circle, inner sep=2pt] (n7) at (4,0) {};
\node[fill=black, draw, circle, inner sep=2pt] (n7) at (4,1) {};
\node[fill=black, draw, circle, inner sep=2pt] (n7) at (3,-1) {};
\node[fill=black, draw, circle, inner sep=2pt] (n7) at (4,-1) {};

 \node[fill=black, draw, circle, inner sep=2pt] (n7) at (2,1) {};
\node[fill=black, draw, circle, inner sep=2pt] (n7) at (5,-1) {};
\node[fill=black, draw, circle, inner sep=2pt] (n7) at (3,-2) {};
\node[fill=black, draw, circle, inner sep=2pt] (n7) at (4,-2) {};
\node[fill=black, draw, circle, inner sep=2pt] (n7) at (5,-2) {};
\node[fill=black, draw, circle, inner sep=2pt] (n7) at (6,-2) {};
\node[fill=black, draw, circle, inner sep=2pt] (n7) at (7,-2) {};

\node[fill=black, draw, circle, inner sep=2pt] (n7) at (4,-3) {};
\node[fill=black, draw, circle, inner sep=2pt] (n7) at (5,-3) {};
\node[fill=black, draw, circle, inner sep=2pt] (n7) at (6,-3) {};
\node[fill=black, draw, circle, inner sep=2pt] (n7) at (7,-3) {};

\node at (-0.2,0.5) {\tiny $s_1$};
\node at (3.5,1.2) {\tiny $s_2$};
\node at (2.8,-1.5) {\tiny $s_3$};
\node at (4.5,-0.8) {\tiny $s_4$};
\node at (3.8,-2.5) {\tiny $s_5$};

\end{tikzpicture}
   \caption{Segments of a polyomino chain.}
    \label{f5}
\end{figure}
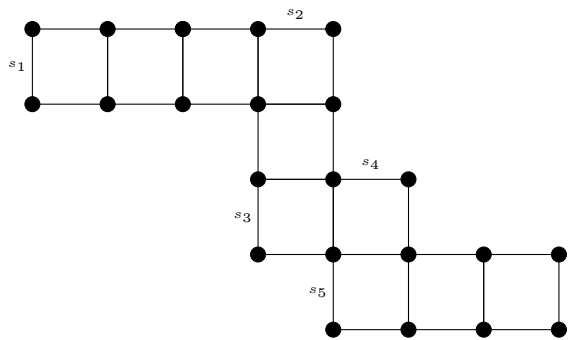

\begin{figure}[h!]
\centering
    \begin{tikzpicture}

  \clip (-3,-1) rectangle (4,3);

\draw (0,0) rectangle (1,1);

\draw (1,0) rectangle (2,1);

\node[fill=white, draw, circle, inner sep=2pt] (n1) at (0,0) {\tiny $v_5$};
\node[fill=white, draw, circle, inner sep=2pt] (n2) at (1,0) {\tiny $v_4$};
\node[fill=white, draw, circle, inner sep=2pt] (n3) at (0,1) {\tiny $v_6$};
\node[fill=white, draw, circle, inner sep=2pt] (n4) at (1,1) {\tiny $v_1$};
\node[fill=white, draw, circle, inner sep=2pt] (n5) at (2,0) {\tiny $v_3$};

\node[fill=white, draw, circle, inner sep=2pt] (n7) at (2,1) {\tiny $v_2$};

  \draw[ dashed] (n1) to  [out=180, in=90,looseness=5] (n4);

\end{tikzpicture}

   \caption{Vertices $v_i$ in $PC_{n-1}$.}
    \label{fn1}
\end{figure}
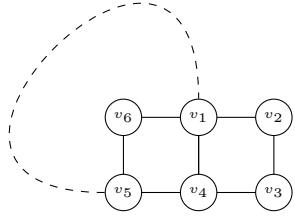

 \begin{theorem}\label{t12}
If $PC_{n}$ is a polyomino chain comprising $n \geq 4$ squares and $m\geq 1$ segment(s),  then
 $$ZC1_{n}= 
      2(16n+10m-2\gamma_1+\gamma_2+I_{\{l_1\geq 4\}}+I_{\{l_m\geq 4\}}-34),
$$
where,  $\gamma_1=\displaystyle\sum_{i=1}^{m}I_{\{l_i=2\}}$ and $\gamma_2=\displaystyle\sum_{i=2}^{m-2}I_{\{l_i=l_{i+1}=2\}}.$

 \end{theorem}
 \begin{proof}
 
\noindent
Given $PC_{n}$, conducting the one-step analysis
 
\begin{eqnarray*}
ZC1_{n}-ZC1_{n-1}=\sum_{v \in V_{n,1}}  \tau_{n,v}^2 + \sum_{v \in V_{n,2}} ( \tau_{n,v}^2-  \tau_{n-1,v}^2),
\end{eqnarray*}
\noindent  
where $ \tau_{k,v}:= \tau_{v}$ is calculated within $PC_{k}$,  $V_{n,1}$ represents the vertices added when progressing from step $n-1$ to $n$, and $V_{n,2}$ denotes the vertices that change their $ \tau_v$ during this advancement. By separately analyzing each sum, we obtain that:

$$\sum_{v \in V_{n,1}}  \tau_{n,v}^2 =8+5I_{\{L_{n}=2\}},$$
and

\begin{align*}
 \sum_{v \in V_{n,2}} (  \tau_{n,v}^2-  \tau_{n-1,v}^2) &=\sum_{v \in V_{n,2}} ( ( \tau_{n-1,v}+1)^2-  \tau_{n-1,v}^2)\\
   &=|V_{n,2}|+2\sum_{v \in V_{n,2}}  \tau_{n-1,v}\\
    &=4+2\sum_{i=1}^{4} \tau_{n-1,v_i}+(2 \tau_{n-1,v_5}+1)I_{\{L_n=2\}}.
\end{align*}

\noindent
where the vertices $v_i$ for $i\in[5]$ are located in the two end squares of the polyomino chain at time $n-1$, as shown in Figure \ref{fn1}. Substituting these values, we derive: 

$$ZC1_{n}=ZC1_{n-1}+12+2( \tau_{n-1,v_5}+3)I_{\{L_n=2\}}+2\sum_{i=1}^{4} \tau_{n-1,v_i}.$$
Now, we need to examine the values of $ \tau_{n-1,v_i}$ for $i\in[5]$. They can be obtained using the above one-step analysis, but now considering the chain in reverse; since they form part of the two end squares of the polyomino chain at time $n-1$. Referring to this, Table \ref{ta1} provides the corresponding values of $ \tau_{n-1,v_i}$ in each entry, i.e, the entry in the $n$-th row and $i$-th column, give us the number of  all the vertices in $PC_{n-1}$ situated at a
distance of two from $v_{i}$, where $v_{i}$ is one of the last six vertices of the $PC_{n-1}$. 

\begin{table}[htbp]
\centering
\begin{tabular}{|c|c|c|c|c|c|c|}
\hline
\multirow{2}{*}{$n$} & \multicolumn{6}{c|}{$i$} \\ \cline{2-7}
                     & 1 & 2 & 3 & 4 & 5 & 6 \\ \hline
3                    & 2 & 2 & 2 & 2 & 2 & 2 \\ \hline
4                    & 3 & $2+I_{\{L_{3}=2\}}$ & 2 & 3 & 3 & 3 \\ \hline
5                    & $4-I_{\{L_{3}=L_{4}=1\}}$ & $2+I_{\{L_{4}=2\}}$ & 2 & 3 & $4-I_{\{L_{3}=1, L_{4}=2\}}$ & 4 \\ \hline
$\geq 6$             & $4-I_{\{L_{n-2}=L_{n-1}=1\}}$ & $2+I_{\{L_{n-1}=2\}}$ & 2 & 3 & $4-I_{\{L_{n-2}=1, L_{n-1}=2\}}$ & $5-I_{\{L_{n-3}=L_{n-2}=1\}}$ \\ \hline
\end{tabular}
\caption{Table with the values of $\tau_{n-1,v_i}$, where $v_{i}$ is one of the last six vertices of the $PC_{n-1}$.}
\label{ta1}
\end{table}

Therefore, for $n\geq 5$

\begin{align} \label{recursiveZ1}
ZC1_n=ZC1_{n-1}+g_1(L_{n-2},L_{n-1},L_{n}),
\end{align}
where $$g_1(L_{n-2},L_{n-1},L_{n})=34+14I_{\{L_{n}=2\}}-2I_{\{L_{n-2}=1,L_{n-1}=L_n=2\}}+2I_{\{L_{n-1}=2\}}-2I_{\{L_{n-2}=L_{n-1}=1\}}.$$
From the information above, applying the recursive formula (\ref{recursiveZ1}) for $n \geq 5$ yields:

\begin{equation}\label{szc}
ZC1_n=ZC1_4 + \sum_{(i,j,k)\in {\{1,2\}}^3} g_{1}(i,j,k)X_{i,j,k},
\end{equation}
where $X_{i,j,k}=|\{w\in\{5,\dots,n\}~:~L_{w-2}=i, L_{w-1}=j$ and $L_{w}=k$ ~in~$PC_{n} \}|$. Thus, upon substituting $n = 4$ into the preceding expression, we confirm the recovery of the initial value $ZC1_4$. Then, note that
\begin{align*}
    X_{1,1,1}+X_{2,1,1} &= X_{1,1}-I_{\{l_1\geq 4\}}, \\
    X_{1,1,2}+X_{2,1,2} &= X_{1,2}-I_{\{l_1=3\}}, \\
    X_{1,2,1}+X_{2,2,1} &= X_{2,1}-I_{\{l_1=2,l_2\neq 2\}}, \\
    X_{1,2,2}+X_{2,2,2} &= X_{2,2}-I_{\{l_1=l_2=2\}}.
\end{align*}
where, similarly, $X_{i,j}=|\{w\in\{4,\dots,n\}~:~L_{w-1}=i$ and $L_{w}=j$ in $PC_{n}\}|$. Now, due to the values of $g_1$ and the above relations, Equation (\ref{szc}) changes to 

\begin{equation}\label{ue1}
\begin{aligned}
ZC1_n &= ZC1_4 + 32(X_{1,1}-I_{\{l_1\geq 4\}})+46(X_{1,2}-I_{\{l_1=3\}})+36(X_{2,1}-I_{\{l_1=2,l_2\neq 2\}})\\
&\quad +48(X_{2,2}-I_{\{l_1=l_2=2\}}) +2X_{2,1,1}+2X_{2,1,2}+2X_{2,2,2}.
\end{aligned}
\end{equation}
Attempting to condense the main expression once more, observe that
\begin{align*}
    X_{2,1,1}+X_{2,1,2} &= X_{2,1}-I_{\{l_m=3\}}, \\
     X_{1,1}+X_{2,1} &= X_{1}-I_{\{l_1\neq 2\}},\\
    X_{2,2}+X_{1,2} &= X_{2}-I_{\{l_1=2\}},\\
    X_{2,1}+X_{2,2} &= X_{2}-I_{\{l_m=2\}},
   \end{align*}
where, $X_{2}=|\{w\in\{3,\dots,n\}~:~L_{w}=2$ in $PC_{n}\}|$ and $X_{1}=|\{w\in\{3,\dots,n\}~:~L_{w}=1$ in $PC_{n}\}|$. After this step, the update Equation (\ref{ue1}) is

\begin{equation}\label{ue}
\begin{aligned}
ZC1_n &= ZC1_4 -32+32X_1+48X_2 -14I_{\{l_2=2\}}-2I_{\{l_m=2\}}-32I_{\{l_1\geq 4\}}-46I_{\{l_1=3\}} \\
&\quad -36I_{\{l_1=2,l_2\neq 2\}}-48I_{\{l_1=l_2=2\}}-2I_{\{l_m=3\}} +4X_{2,1} +2X_{2,2,2}.
\end{aligned}
\end{equation}
Let us review the unknown variables: if at time $w$, $L_{w}=2$ then the last segment in $PC_{w-1}$ is completed (initiating a new segment in $PC_{w}$). Conversely, if at time $w$, $L_{w}=1$, then a square is added to the last segment in $PC_{w-1}$. Hence, $X_{2}=m-1$, $X_{1}=n-m-1$, $X_{2,1}=m-\gamma_1-1+I_{\{l_1=2\}}$ and $X_{2,2,2}=\gamma_2$. To determine $ZC1_4$, we can use the same recursive expression, considering the values of $ \tau_{3,v_i}$. Thus,
\begin{equation*}
\begin{aligned}
ZC1_4 &= 84+12I_{\{L_4=2\}}+12I_{\{L_3=2\}} \\
&=84+12I_{\{l_1=3\}}+12I_{\{l_1=2,l_2\neq 2\}}+24I_{\{l_1=l_2=2\}}.
\end{aligned}
\end{equation*}

Finally, substituting in (\ref{ue}) the values found and properly rewriting the excess terms, we obtain the result.

\end{proof}

\begin{theorem}\label{t13}
If $PC_{n}$ is a polyomino chain comprising $n \geq 5$ squares and $m\geq 1$ segment(s),  then

\begin{align*}
ZC2_{n} &= 48n+34m-9\gamma_1-6\sum_{\scriptscriptstyle i=1,m}I_{\{l_i=2\}}-4\sum_{\scriptscriptstyle i=1,m}I_{\{l_i=3\}}  -2\sum_{\scriptscriptstyle i=1,m}I_{\{l_i=4\}}-\sum_{\scriptscriptstyle i=1,m-2}I_{\{l_i=l_{i+1}=l_{i+2}=2,m\geq 3\}} \\
&\quad +I_{\{l_1=2,l_2 \geq 4,m\geq 2\}}+I_{\{l_m=2,l_{m-1} \geq 4,m\geq 2\}}+\gamma_3+\gamma_4+\gamma_5I_{\{m\geq 4\}}-112,
\end{align*}
where,  $\gamma_1=\displaystyle\sum_{i=1}^{m}I_{\{l_i=2\}}$, $\gamma_3=\displaystyle\sum_{i=2}^{m-2}I_{\{l_{i+1}=3 \lor l_i=l_{i+1}=2 \lor l_i,l_{i+1}\neq 2 \lor l_i=3\}}$, $\gamma_4=\displaystyle\sum_{i=1}^{m}I_{\{l_i=4\}}$ and $\gamma_5=2\displaystyle\sum_{i=3}^{m-2}I_{\{l_i=2\}}-m+I_{\{l_2=2\}}+I_{\{l_{m-1}=2\}}+3$.

 \end{theorem}
 \begin{proof}
 
\noindent
By conducting the same analysis at one step, we conclude that
 
\begin{eqnarray*}
ZC2_{n}-ZC2_{n-1}=\sum_{uv \in E_{n,1}}  f_{n}(uv)+ \sum_{uv \in E_{n,2}} (f_{n}(uv)- f_{n-1}(uv)),
\end{eqnarray*}
\noindent  
where $f_{k}(uv):= \tau_{k,u} \tau_{k,v}$, $E_{n,1}$ are the edges that are added as we advance from step $n-1$ to $n$ and $E_{n,2}$ are the edges that modify their $\tau_{u}\tau_{v}$ as we advance from step $n-1$ to $n$. If we focus on each sum individually, it turns out that for $n \geq 6$,

$$\sum_{uv \in E_{n,1}}  f_{n}(uv)=16+2I_{\{L_{n-1}=2,L_{n}=1\}}+8I_{\{L_{n}=2\}},$$
and
\begin{align*} 
&\displaystyle\sum_{uv \in E_{n,2}} (f_{n}(uv)- f_{n-1}(uv))= 4+2\sum_{i=1}^{4}  \tau_{n-1,v_i}+ \sum_{i=5}^{6}  \tau_{n-1,v_i}+( \tau_{n-1,v_4}+ \tau_{n-1,v_6}+1)I_{\{L_n=2\}} \\ &~~~~~~~~~~~~~~~~~~~~~~~~~~~~~~~~~~~~~~~~~~+(4-I_{\{L_{n-3}=1,L_{n-2}=2\}})(I_{\{L_{n-1}=1,L_{n}=2\}}+I_{\{L_{n-1}=2\}})+I_{\{L_{n-1}=2\}},  
\end{align*}
where the vertices $v_i$ with $i\in[5]$ are those used previously and the vertex $v_{6}$ is also presented in Figure \ref{fn1}. Then, substituting for $n \geq 6$, we obtain that 
\begin{align}\label{recursive-Z2}
ZC2_{n}=ZC2_{n-1}+g_2(L_{n-3},L_{n-2},L_{n-1},L_{n}),
\end{align}
where
\begin{equation*}
\begin{split}
g_2(L_{n-3},L_{n-2},L_{n-1},L_{n})= & 51+7I_{\{L_{n-1}=2\}}+17I_{\{L_{n}=2\}}-I_{\{L_{n-2}=1,L_{n-1}=2\}} -I_{\{L_{n-3}=L_{n-2}=1\}}-2I_{\{L_{n-2}=L_{n-1}=1\}} \\
& +2I_{\{L_{n-1}=2,L_{n}=1\}}+4I_{\{L_{n-1}=1,L_{n}=2\}} -I_{\{L_{n-3}=L_{n-1}=1, L_{n-2}=L_n=2\}}-I_{\{L_{n-3}=1,L_{n-2}=L_{n-1}=2\}} \\
& -I_{\{L_{n-3}=L_{n-2}=1,L_n=2\}}.
\end{split}
\end{equation*}
Now, by applying the recursive formula (\ref{recursive-Z2}) for $n \geq 6$, it follows that:

\begin{equation}\label{szc1}
ZC2_n=ZC2_5 + \sum_{(i,j,k,l)\in {\{1,2\}}^4} g_{1}(i,j,k,l)X_{i,j,k,l},
\end{equation}
where $X_{i,j,k,l}=|\{w\in\{6,\dots,n\}~:~L_{w-3}=i, L_{w-2}=j,L_{w-3}=k$ and $L_{w}=l$~in ~$PC_n\}|$. Verifying that upon substituting $n$ with 5 in the previous expression, we retrieve the initial value $ZC2_5$. Additionally, 
\begin{align*}
    X_{1,1,1,1}+X_{2,1,1,1} &= X_{1,1,1}-I_{\{l_1\geq 5\}}, \\
    X_{1,1,1,2}+X_{2,1,1,2} &= X_{1,1,2}-I_{\{l_1=4\}}, \\
    X_{1,1,2,1}+X_{2,1,2,1} &= X_{1,2,1}-I_{\{l_1=3,l_2\neq 2\}}, \\
    X_{1,1,2,2}+X_{2,1,2,2} &= X_{1,2,2}-I_{\{l_1=3,l_2=2\}},\\
      X_{1,2,1,1}+X_{2,2,1,1} &= X_{2,1,1}-I_{\{l_1=2,l_2\geq 4\}}, \\
    X_{1,2,1,2}+X_{2,2,1,2} &= X_{2,1,2}-I_{\{l_1=2,l_2=3\}}, \\
    X_{1,2,2,1}+X_{2,2,2,1} &= X_{2,2,1}-I_{\{l_1=l_2=2,l_3\neq 2\}}, \\
    X_{1,2,2,2}+X_{2,2,2,2} &= X_{2,2,2}-I_{\{l_1=l_2=l_3=2\}}.\\
     X_{2,1,1,1}+X_{2,1,1,2} &= X_{2,1,1}-I_{\{l_m=4\}}, \\
    X_{2,1,2,1}+X_{2,1,2,2} &= X_{2,1,2}-I_{\{l_{m-1}=3,l_m=2\}}, \\
    X_{2,2,2,1}+X_{2,2,2,2} &= X_{2,2,2}-I_{\{l_{m-2}=l_{m-1}=l_{m-2}=2\}},\\
     X_{2,2,1}+X_{2,2,2} &= X_{2,2}-I_{\{l_{m-1}=l_{m}=2\}}.
\end{align*}

Therefore, using the above relations, together with the similar ones employed in Theorem \ref{t12} concerning $X_{i,j,k}$ and $X_{i,j}$, it all boils down to expressing in terms of segment lengths $X_{2,1,1,2}$, $X_{2,1,2,2}$, $X_{2,2,1,2}$, and $ZC2_5$. Finally, by performing some algebraic manupilations with the excess terms, the demonstration is completed.

\end{proof}

To complement the preceding theorems, Table \ref{table1} presents the values of the analyzed indices for the cases $n\leq6$.

\begin{table}[!htbp]
\begin{center}
\begin{tabular}{ |c|l|c|c| }
\hline
$n$ & Links & $ZC1(PC_n)$ & $ZC2(PC_n)$ \\
\hline
2 & & 24 & 28 \\ 
\hline
\multirow{2}{*}{3} & 1 & 52 & 68 \\ 
                   & 2 & 62 & 78 \\ 
\hline
\multirow{3}{*}{4} & (1,1) & 84 & 114\\ 
                   & (1,2) and (2,1) & 96 & 131 \\ 
                   & (2,2) & 108 & 144 \\ 
\hline
\multirow{6}{*}{5} & (1,1,1) & 116 & 162 \\ 
                   & (1,1,2) and (2,1,1) & 130  & 181 \\ 
                   & (1,2,1) & 132 & 188 \\ 
                   & (2,1,2) & 144 & 200 \\ 
                   & (1,2,2) and (2,2,1) & 144 & 202 \\ 
                   & (2,2,2) & 158 & 216 \\ 
\hline
\multirow{10}{*}{6} & (1,1,1,1) & 148 & 210 \\ 
                    & (1,1,1,2) and (2,1,1,1) & 162  & 230 \\ 
                    & (1,1,2,1) and (1,2,1,1) & 166 & 239 \\ 
                    & (2,1,1,2) & 176 & 251 \\ 
                    & (1,1,2,2) and (2,2,1,1) & 178 & 253 \\ 
                    & (1,2,1,2) and (2,1,2,1) & 180 & 259 \\ 
                    & (1,2,2,1) & 180 & 261 \\ 
                    & (2,1,2,2) and (2,2,1,2) & 192 & 274 \\ 
                    & (1,2,2,2) and (2,2,2,1) & 194 & 276 \\ 
                    & (2,2,2,2) & 208 & 291 \\ 
\hline
\end{tabular}
\end{center}
\caption{The values of $ZC1_n$ and $ZC2_n$ for $n \leq 6$.}
\label{table1}
\end{table}

\noindent
 Now, we will utilize the main results to compute $ZC1_n$ and $ZC2_{n}$ for certain types of polyomino chains studied in the literature \cite{bd35,bd34,bd33}.

\begin{corollary}\label{R1}
For the polyomino chain with $n$ squares and two segments $s_{1}$ and $s_{2}$ satisfying that $l_{1}=2$
and $l_{2}=n-1$, denoted as $PC_{n}^{1}$, we have the following:

$$ZC1_n=32n-30-2I_{\{n=4\}}, ~n \geq 4,$$
$$ZC2_n=48n-58-I_{\{n=5\}}, ~n \geq 5.$$

\end{corollary}

\begin{corollary}\label{R2}
For the polyomino chain with $n$ squares and  $m \geq 3$ segments $s_{1}, s_{2}, \dots, s_{m}$ satisfying that $l_{1}=l_{m}=2$
and $l_{2},\dots,l_{m-1}\geq 3$, denoted as $PC_{n}^{2}$, we have the following:

$$ZC1_n=32n+20m-76, ~n \geq 4,$$
$$ZC2_n=48n+34m-140-I_{\{l_2=3\}}-I_{\{l_{m-1}=3\}}+m_4, ~n \geq 5,$$
where $m_4$ denotes the number of segments of length equal to 4.
\end{corollary}

\noindent
On the other hand, we denote a polyomino chain of dimension $n\geq 1$ with $k = k_{1}+ k_{2} + k_{3}$, where $k_{1}$ is the number of kinks, $k_{2}$ is the number of medials, and $k_{3}$ is the number of terminals in a unit of polyomino chain, by $PC_{n,k}$. Figure \ref{f3} illustrates a general representation of a polyomino chain $PC_{n,k}$. Let $k\geq 3$, by definition of $PC_{n,k}$, we have: $m=2n,$ $\gamma=n-1$, $I_{\{l_m=2\}}=1$, and $I_{\{l_1=2\}}=I_{\{k=3\}}$. Hence, in the following corollary, we will compute $ZC1(PC_{n,k})$ and $ZC2(PC_{n,k})$ for $k\geq 3$ using Theorems \ref{t12} and \ref{t13}.

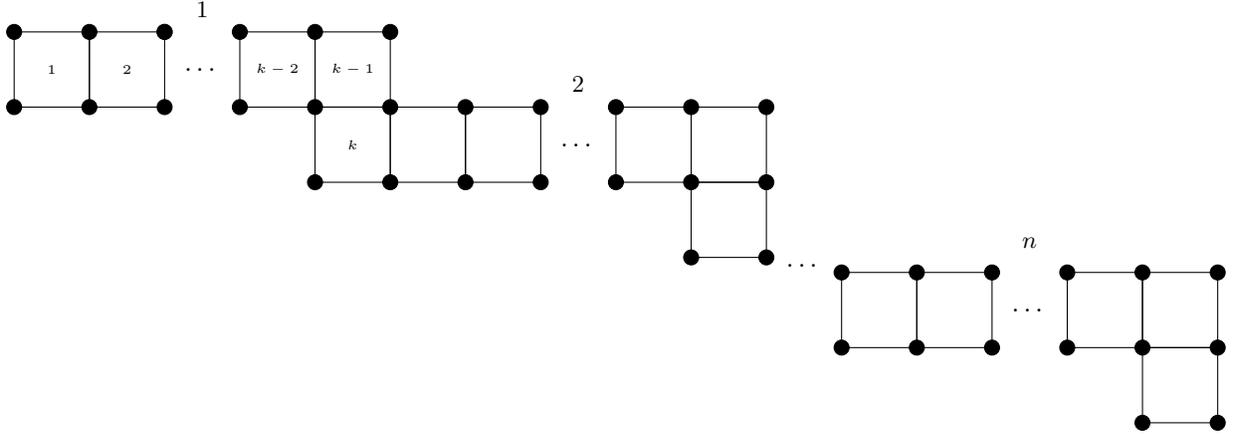
\begin{figure}[h!]
    \centering
   
\begin{tikzpicture}

\draw (0,0) rectangle (1,1);
\draw (1,0) rectangle (2,1);
\draw (4-1,0) rectangle (5-1,1);
\draw (6-2,0) rectangle (7-2,1);
\draw (6-2,-1) rectangle (7-2,0);
\draw (7-2,-1) rectangle (8-2,0);
\draw (8-2,-1) rectangle (9-2,0);
\draw (11-3,-1) rectangle (12-3,0);
\draw (12-3,-1) rectangle (13-3,0);
\draw (9,-2) rectangle (13-3,-1);

\draw (11,-3.2) rectangle (12,-2.2);
\draw (12,-3.2) rectangle (13,-2.2);
\draw (14,-3.2) rectangle (15,-2.2);

\draw (15,-3.2) rectangle (16,-2.2);
\draw (15,-4.2) rectangle (16,-3.2);

\node[fill=black, draw, circle, inner sep=2pt] (n1) at (5,1) {};
\node[fill=black, draw, circle, inner sep=2pt] (n1) at (5,0) {};
\node[fill=black, draw, circle, inner sep=2pt] (n1) at (6,-1) {};
\node[fill=black, draw, circle, inner sep=2pt] (n1) at (6,0) {};
\node[fill=black, draw, circle, inner sep=2pt] (n1) at (7,-1) {};
\node[fill=black, draw, circle, inner sep=2pt] (n1) at (7,0) {};

\node[fill=black, draw, circle, inner sep=2pt] (n1) at (8,-1) {};
\node[fill=black, draw, circle, inner sep=2pt] (n1) at (8,0) {};
\node[fill=black, draw, circle, inner sep=2pt] (n1) at (9,-1) {};
\node[fill=black, draw, circle, inner sep=2pt] (n1) at (9,0) {};
\node[fill=black, draw, circle, inner sep=2pt] (n1) at (10,-1) {};
\node[fill=black, draw, circle, inner sep=2pt] (n1) at (10,0) {};
\node[fill=black, draw, circle, inner sep=2pt] (n1) at (9,-2) {};
\node[fill=black, draw, circle, inner sep=2pt] (n1) at (10,-2) {};

\node[fill=black, draw, circle, inner sep=2pt] (n1) at (11,-3.2) {};
\node[fill=black, draw, circle, inner sep=2pt] (n1) at (11,-2.2) {};

\node[fill=black, draw, circle, inner sep=2pt] (n1) at (12,-3.2) {};
\node[fill=black, draw, circle, inner sep=2pt] (n1) at (12,-2.2) {};

\node[fill=black, draw, circle, inner sep=2pt] (n1) at (13,-3.2) {};
\node[fill=black, draw, circle, inner sep=2pt] (n1) at (13,-2.2) {};

\node[fill=black, draw, circle, inner sep=2pt] (n1) at (14,-3.2) {};
\node[fill=black, draw, circle, inner sep=2pt] (n1) at (14,-2.2) {};

\node[fill=black, draw, circle, inner sep=2pt] (n1) at (15,-3.2) {};
\node[fill=black, draw, circle, inner sep=2pt] (n1) at (15,-2.2) {};

\node[fill=black, draw, circle, inner sep=2pt] (n1) at (16,-3.2) {};
\node[fill=black, draw, circle, inner sep=2pt] (n1) at (16,-2.2) {};

\node[fill=black, draw, circle, inner sep=2pt] (n1) at (15,-4.2) {};
\node[fill=black, draw, circle, inner sep=2pt] (n1) at (16,-4.2) {};

\node[fill=black, draw, circle, inner sep=2pt] (n1) at (0,0) {};
\node[fill=black, draw, circle, inner sep=2pt] (n2) at (1,0) {};
\node[fill=black, draw, circle, inner sep=2pt] (n3) at (0,1) {};
\node[fill=black, draw, circle, inner sep=2pt] (n4) at (1,1) {};
\node[fill=black, draw, circle, inner sep=2pt] (n5) at (2,0) {};

\node[fill=black, draw, circle, inner sep=2pt] (n7) at (2,1) {};
\node[fill=black, draw, circle, inner sep=2pt] (n7) at (3,0) {};
\node[fill=black, draw, circle, inner sep=2pt] (n7) at (3,1) {};
\node[fill=black, draw, circle, inner sep=2pt] (n7) at (4,0) {};
\node[fill=black, draw, circle, inner sep=2pt] (n7) at (4,1) {};

\node[fill=black, draw, circle, inner sep=2pt] (n7) at (4,-1) {};

 \node[fill=black, draw, circle, inner sep=2pt] (n7) at (2,1) {};
\node[fill=black, draw, circle, inner sep=2pt] (n7) at (5,-1) {};

\node at (0.5,0.5) {\tiny $1$};
\node at (1.5,0.5) {\tiny $2$};
\node at (3.5,0.5) {\tiny $k-2$};
\node at (4.5,0.5) {\tiny $k-1$};
\node at (4.5,-0.5) {\tiny $k$};

\node at (2.5,0.5) {$\dots$};
\node at (2.5,1.3) {\small $1$};
\node at (7.5,-0.5) {$\dots$};
\node at (7.5,0.3) {\small $2$};
\node at (10.5,-2.1) {$\dots$};
\node at (13.5,-2.7) {$\dots$};
\node at (13.5,-1.8) {\small $n$};

\end{tikzpicture}

   \caption{General representation of $PC_{n,k}$.}
    \label{f3}
\end{figure}
\begin{remark}
  Note that, by definition $PC_{n,1}=Li_{n}$ and $PC_{n,2}=Z_{2n}$.
 \end{remark}
\begin{corollary}\label{R5}
 
Given $k \geq 3$, $n\geq 2$, it follows that

  $$ZC1(PC_{n,k})=n(32k+36)-66-6I_{\{k=3\}}-2I_{\{k=4\}}.$$
   $$ZC2(PC_{n,k})=n(48k+59)-117+(2n-19)I_{\{k=3\}}+(n-5)I_{\{k=4\}}-I_{\{k=5\}}.$$

\end{corollary}

\noindent
By definition, if $PC_{n}=Li_{n}$, we deduce that $m=1$ and $l_{1}=n$. Similarly, if $PC_{n}=Z_{n}$, then $m=n-1$ and $l_{i}=2$ for $i=1,2,\dots,m$. Therefore, the following results can be obtained by applying the previous theorems.

\begin{proposition}
    \label{c1}
Let $Li_{n}$ and $Z_{n}$ be linear and zigzag chains, respectively, with $n \geq 3$ squares. Then

$$ZC1(Li_n) \leq  ZC1(PC_n) \leq ZC1(Z_n).$$
Equality is attained if and only if $PC_n=Li_n$ and $PC_n= Z_n$, respectively. Moreover, $ZC1(Li_n)=32n-44$ for $n \geq 3$ and $ZC1(Z_n)=50n-92$ for $n \geq 4$.

\end{proposition}
\begin{proof}
    The proof will proceed by induction on $n$. For $n \leq 6$, verification will be conducted based on Table \ref{table1}. Now, assuming that the assertion holds for $n-1$, i.e.,
    $$ZC1(Li_{n-1}) \leq ZC1_{n-1} \leq ZC1(Z_{n-1}),$$ and that equality is reached if and only if $PC_{n-1}=Li_{n-1}, Z_{n-1}$, respectively. By utilizing the recursive formula provided in Theorem \ref{t12}, for $n \geq 6$, we can assert that $ZC1_n=ZC1_{n-1}+g_1(L_{n-2},L_{n-1},L_n)$. Then, it can be verified that
   $g_1(1,1,1)=32\leq g_1(L_{n-2},L_{n-1},L_n) \leq 50=g_1(2,2,2)$.
   Moreover, the equality holds if and only if $L_{n-2}=L_{n-1}=L_n=1$ and $L_{n-2}=L_{n-1}=L_n=2$, respectively. Thus, 
    $$ZC1(Li_{n-1})+32=ZC1(Li_{n})\leq ZC1_n \leq ZC1(Z_{n})=ZC1(Z_{n-1})+50.$$
    Concluding the proof.
   
\end{proof}
Simply by working with $g_2(L_{n-3},L_{n-2},L_{n-1},L_n)$ as it appears in the recursive formula presented in Theorem \ref{t13}, a similar line of reasoning allows us to prove the following proposition.

\begin{proposition}\label{c2}

Let $Li_{n}$ and $Z_{n}$ be linear and zigzag chains, respectively, with $n \geq 3$ squares. Then

$$ZC2(Li_n) \leq  ZC2(PC_n) \leq ZC2(Z_n).$$
Equality is attained if and only if $PC_n=Li_n$ and $PC_n=Z_n$, respectively. Moreover, $ZC2(Li_n)=48n-78$ for $n \geq 4$ and $ZC2(Z_n)=75n-159$ for $n \geq 5$. 

\end{proposition}

\begin{remark}\label{re1}
   This remark aims to underscore certain implications of Theorem \ref{t12} and Theorem \ref{t13}, which are derived similarly through the recursive formula and properties of the functions $g_1$ and $g_2$.
\begin{itemize}
\item[(a)] Given that, for a fixed $PC_n$, $g_1(L_{n-2},L_{n-1},L_n) < g_2(L_{n-3},L_{n-2},L_{n-1},L_n)$, it follows that for $n \geq 2$, $ZC1_n < ZC2_n$. 
        \item[(b)]  Because of the sequence of values of the function $g_1(g_2)$, substituting a  link of type 1 with a link of type 2 in a fixed $PC_n$ leads to an increase in the index. After this, one might think that this property suggests that when the number of links equal to two ($X_2$) increases then the index increases. However, this is generally not true, see the following examples written with the notation introduced in the paragraph below: $$ZCi(PC(2,2,2,1,1,1,1,1,1,2,2)) <  ZCi(PC(1,1,2,1,2,1,2,1,2,1,1)),~i=1,2$$

        $$ZC1(PC(1,2,1,2,1,2,1,2,1,1))=ZC1(PC(2,2,1,1,1,1,1,2,2,2))$$

        and

        $$ZC2(PC(1,2,1,1,1,2,1,1,1,1,1,1))=ZC2(PC(2,2,1,1,1,1,1,1,1,1,1,2)).$$
    \end{itemize}
\end{remark}

The aforementioned propositions indicate where the maximum and minimum are attained. Furthermore,  Remark \ref{re1} (b) enables us to compare certain types of chains. Similarly, the subsequent results provide insight into how the indices behave beyond the extremes. To facilitate the interpretation of the results, we will adopt the following notation: $PC(L_3, L_4, \dots, L_n)$ represents a polyomino chain with $n$ squares, with the links chosen at each stage according to the vector $v=(L_3, L_4, \dots, L_n)$. In the forthcoming two results, we will primarily focus on polyomino chains containing either one link of type 1 or one link of type 2 and the number of links preceding either the link of type 1 or the link of type 2 will be greater than or equal to $ \lfloor (n-2)/2 \rfloor$.

\begin{proposition}\label{p1d}
   Given $PC_n$ with $n \geq 5$ squares. Then
   $$ZC1(Li_n) < ZC1(PC(1, 1, \dots ,1, 2))< ZC1(PC(1, 1, \dots, 1,2, 1))$$
   $$< ZC1(PC(1, 1, \dots, 1,2,1, 1))= ZC1(PC(\underbrace{1, 1, \dots, 1}_{i},2,1, \dots, 1,1)),$$
and 
$$ZC1(Z_n) > ZC1(PC(2, 2, \dots ,2, 1)) >ZC1(PC(2, 2, \dots, 2,1, 2))$$
 $$> ZC1(PC(2, 2, \dots, 2,1,2, 2))= ZC1(PC(\underbrace{2, 2, \dots, 2}_{i},1,2, \dots, 2,2))$$
with $\lfloor (n-2)/2 \rfloor \leq i \leq n-6$. Moreover, for $n \geq 6$ 
$$ZC1(PC(1, 1, \dots, 1,2,1, 1)) < ZC1(PC(2, 1, \dots, 1, 2)),$$
and

$$ZC1(PC(2, 2, \dots, 2,1,2, 2)) >ZC1(PC(1, 2, \dots, 2, 1)).$$
\end{proposition}
\begin{proof}
    Note that, for $n=5,6$, the result is verified according to Table \ref{table1}. Therefore, we will prove it for $n \geq 7$. The first inequality follows from  Proposition \ref{c1} (also by  Remark \ref{re1} (b)). For the second inequality, observe that, by symmetry, $PC(1, 1, \dots ,1, 2)=PC(2, 1, \dots ,1, 1)$. Thus, by applying the recursive formula presented in Theorem \ref{t12}, we conclude that:

    $$ZC1(PC(2,\underbrace{ 1, \dots ,1, 1}_{n-3}))=ZC1(PC(2,\underbrace{ 1, \dots ,1}_{n-4}))+g_1(1,1,1),$$
and
    $$ZC1(PC(1, 1, \dots, 1,2, 1))=ZC1(PC(1, 1, \dots, 1,2))+g_1(1,2,1).$$
Given that the first summands are equal by symmetry and $g_1(1,1,1)< g_1(1,2,1)$, the required result follows. The third inequality can be verified through a similar analysis, thus, we will proceed to prove the final part. To do this, we will demonstrate that:
    
    $$ZC1(PC(\underbrace{1, 1, \dots, 1}_{i},2,1, \dots, 1,1))=ZC1(PC(\underbrace{1, 1, \dots, 1}_{i-1},2,1, \dots, 1,1)),$$
    with $\lfloor (n-2)/2 \rfloor +1 \leq i \leq n-5$. In this case,  rewriting the expressions as follows, 

 $$ZC1(PC(\underbrace{1, 1, \dots, 1}_{i},2,1, \dots, 1,1))=ZC1(PC(1, 1, \dots, 1,2,\underbrace{1, \dots, 1,1}_{i-1}))+g_1(1,1,1),$$
 and
$$ZC1(PC(\underbrace{1, 1, \dots, 1}_{i-1},2,1, \dots, 1,1))=ZC1(PC(\underbrace{1, 1, \dots, 1}_{i-1},2,1, \dots, 1,1))+g_1(1,1,1),$$   
since,  $i\geq 3$ and $n-2-i\geq 3$, we conclude the equalities. For the decreasing inequalities, we proceed in a similar manner, considering that $g_1(2,2,2)$ is the maximum. As for the final inequalities, which involve a chain with 2 links of type 2 and 2 links of type 1 respectively, by iterating the recursive formula, it is demonstrated that:

\begin{align*}
& ZC1(PC(2, 1, \dots, 1, 2)) - ZC1(PC(1, 1, \dots, 1, 2, 1,1)) \\
& = g_1(1, 1, 1) + g_1(1, 1, 2) - g_1(1, 2, 1) - g_1(2, 1, 1).
\end{align*}
and
\begin{align*}
& ZC1(PC(1, 2, \dots, 2, 1)) - ZC1(PC(2, 2, \dots, 2,1,2,2) \\
& = g_1(2,2,2) + g_1(2,2,1) - g_1(2,1,2) - g_1(1,2,2).
\end{align*}
Now, since $g_1(1,1,1)+g_1(1,1,2)-g_1(1,2,1)-g_1(2,1,1)> 0$ and $g_1(2,2,2)+g_1(2,2,1)-g_1(2,1,2)-g_1(1,2,2) < 0$, the proof is completed.
\end{proof}

Using a similar approach, we obtain the following proposition.

\begin{proposition}\label{p2d}
   Given $PC_n$ with $n \geq 5$ squares. Then
   $$ZC2(Li_n) < ZC2(PC(1, 1, \dots ,1, 2))< ZC2(PC(1, 1, \dots 1,2, 1))< ZC2(PC(1, 1, \dots 1,2,1, 1))$$
   $$< ZC2(PC(1, 1, \dots 1,2,1,1, 1))= ZC2(PC(\underbrace{1, 1, \dots 1}_{i},2,1, \dots, 1,1)),$$
and
 $$ZC2(Z_n) > ZC2(PC(2, 2, \dots ,2, 1)) >ZC2(PC(2, 2, \dots 2,1, 2))> ZC2(PC(2, 2, \dots 2,1,2, 2))$$
 $$> ZC2(PC(2, 2, \dots 2,1,2, 2))= ZC2(PC(\underbrace{2, 2, \dots 2}_{i},1,2, \dots, 2,2)),$$
with $\lfloor (n-2)/2 \rfloor \leq i \leq n-7$. Moreover, for $n\geq 8$
$$ZC2(PC(1, 1, \dots 1,2,1,1,1))< ZC2(PC(2, 1, \dots, 1, 2)),$$
and

$$ZC2(PC(2, 2, \dots 2,1,2,2,2)) >ZC2(PC(1, 2, \dots, 2, 1)).$$
   
\end{proposition}

\begin{remark}
Similarly, it is possible to show that, for $n\geq 6$ and $i=1,2$

$$ZCi(PC(1,2,\dots,2,1)) > ZCi(PC(2,1,\dots,1,2)).$$
Hence, combining this with the results detailed in the previous propositions, the  ascending and descending sequences in both indices can be connected for $n \geq 5$. On the other hand, when comparing Proposition \ref{p1d} with Proposition \ref{p2d}, it is noted that the latter includes an additional inequality. This is primarily attributed to the fact that the function $g_2$ involves four variables, whereas the function $g_1$ comprises three. In fact, as per the findings presented in \cite{bd14}, the recursive formula for the degree-based topological indices on the $PC_n$ is governed by a function of two variables. Consequently, the question arises as to when similar behaviors persist and when distinct patterns emerge for different degree-based topological indices. Finally, it is crucial to note, although it is not explicitly stated in the above propositions, that due to symmetry
   $$ PC(\underbrace{1, 1, \dots, 1}_{i},2,\underbrace{1, \dots, 1,1}_{n-i-3}))=PC(\underbrace{1, 1, \dots 1}_{n-i
   -3},2,\underbrace{1, \dots, 1,1}_{i})),  $$
   and
      $$ PC(\underbrace{2, 2, \dots, 2}_{i},1,\underbrace{2, \dots, 2,2}_{n-i-3}))=PC(\underbrace{2, 2, \dots, 2}_{n-i-3},1,\underbrace{2, \dots, 2,2}_{i})). $$
      This was the reason why the prepositions were raised only for $i \geq \lfloor (n-2)/2 \rfloor$.
   
\end{remark}
To close this section we would like to highlight the following. Inspired by the above results other questions naturally arise: What happens to the behavior of the indices when the number of links of type 2 is not only 1 and $n-1$ (similarly for the number of links of type 1), and Can we identify any patterns?  Given $n$ and $X_2$ fixed or equivalently given $n$ and $m$ fixed, it can be shown by focusing on the non-fixed summands of the expressions obtained in Theorem \ref{t12} and Theorem \ref{t13} the following: The minimum value of $ZCi$ with $i=1,2$ is 

  $$ZCi(PC(\underbrace{2, \dots ,2}_{ \lfloor X_{2}/2 \rfloor },1,\dots,1,\underbrace{2, \dots, 2}_{\lceil X_{2}/2 \rceil})).$$
  In turn, the maximum value of $ZC1$ is,

  $$ZC1(PC(\underbrace{1,\dots ,1}_{2+\lfloor \frac{n-2X_{2}-5}{2} \rfloor} ,\underbrace{2,1,\dots,2,1}_{2X_2}\underbrace{1,\dots ,1}_{1+\lceil \frac{n-2X_{2}-5}{2} \rceil})),$$
if $(n-5)/2 \geq X_2$ and the maximum value of $ZC2$ is,

  $$ZC2(PC(\underbrace{1,1,\dots ,1}_{3+\lfloor \frac{n-3X_{2}-6}{2} \rfloor} ,\underbrace{2,1,1,\dots,2,1,1}_{3X_2}\underbrace{1,\dots,1,1}_{1+\lceil \frac{n-3X_{2}-6}{2} \rceil})),$$
  if $(n-6)/3 \geq X_2$. Concurrently, other relevant question would be to examine the possible patterns that arise when analyzing the behavior between the minimum and maximum value of all polyomino chains, keeping the value of $n$ and $X_2$ fixed. Regarding this point, we intend to develop it in the future.

\medskip

\section{Random Polyomino Chain}\label{s2}
\noindent

 In this section, we introduce and demonstrate our main results within a randomized framework, offering insights into long-term behavior. We will employ the symbols $\stackrel{a.s}{\longrightarrow}$ and $\stackrel{D}{\longrightarrow}$ to denote convergence in almost surely and distribution, respectively. Here, $N\left(\mu, \sigma^{2}\right)$ represents a random variable following a normal distribution with mean $\mu$ and variance $\sigma^{2}$, and $Mult(n,p)$ represents a random variable following a multinomial distribution with parameters $n$ and $p$.

\begin{theorem}\label{t1s2}
As $n \rightarrow \infty$,
\hfill
\medskip

\begin{center}
    $\frac{ZC1_{n}-(50-18p+2p^2-2p^3)n}{\sqrt{n}} \stackrel{D}{\longrightarrow}N(0,\sigma^{2})$,
\end{center}
where $\sigma^2=4p\left(1-p\right)(5p^4-7p^3+56p^2-35p+81).$
\end{theorem}
\begin{proof}
 Based on the recursive relation found in Theorem \ref{t12}, we can express $ZC1_n$ for $n\geq 4$ as follows
 $$ZC1_n=ZC1_4+\sum_{i=3}^{n-2}G_1(Y_{i}),$$
    where $Y_{i}=L_iL_{i+1}L_{i+2}$, i.e., it is the random variable that gives me how the last three links used up to the time $i + 2$ have been appearing and $G_1(abc):=g_1(a,b,c)$ for $abc \in S:=\{111,112,\dots,221,222\}.$ According to the definition, we have that $\{Y_{i}\}_{i\geq 3}$ corresponds to a homogeneous Markov Chain with transition probability as follows:

$$\mathbb{P}(Y_{4}=def~\mid~Y_3=abc)=\begin{cases}
        p^{2-f}(1-p)^{f-1} & \text{if }  b=d, c=e \\
          
          0 &  \text{otherwise } 
         
    \end{cases},$$
with $abc,def \in S.$ As a first consequence of the probability of transition we have that: any pair of states can be linked within three steps with a positive probability, since 
\begin{align*}
\mathbb{P}(Y_{6}=def \mid Y_3=abc)
&= \mathbb{P}(Y_{4}=bcd \mid Y_{3}=abc)\mathbb{P}(Y_{5}=cde \mid Y_4=bcd)\mathbb{P}(Y_{6}=def \mid Y_5=cde)\\
&=  p^{6-(d+e+f)}(1-p)^{(d+e+f)-3}>0.
\end{align*}
Thus, each entry of $P^3$ is positive, where $P$ denotes the transition matrix of the  homogeneous Markov Chain $\{Y_i\}_{i\geq 3}$.  In this way, the properties of irreducibility and aperodicity are proved simultaneously. At this point based on the Perron-Frobenius Theorem \cite{bd421,bd422}, we have that the limit distribution exists. On the other hand, the initial distribution $\pi_3$, verifies that 
\begin{align*}
\pi_3(def)&=\mathbb{P}(Y_3=def)\\
&= p^{6-(d+e+f)}(1-p)^{(d+e+f)-3}.
\end{align*}
With this clarification, we can see that another implication of the transition probability is that the initial distribution $\pi_3$ is the limiting distribution, because

\begin{align*}
\sum_{abc \in S}\pi_3(abc)\mathbb{P}(Y_{4}=def~\mid~Y_3=abc) &\quad=p^{2-f}(1-p)^{f-1}\sum_{a \in \{1,2\}}\pi_3(ade) \\
&\quad =p^{6-(d+e+f)}(1-p)^{(d+e+f)-3}\\
&\quad =\pi_3(def).
\end{align*}
 Now, considering the measurable real-valued function $G_1$ for which $\mathbb{V}(G_1(Y_3)) < +\infty $, by virtue of the Markov Chain Central Limit Theorem \cite{bd1,bd2}, we deduce that

$$\frac{ZC1_n-\mathbb{E}(G_1(Y_3))n}{\sqrt{n}}\stackrel{D}{\longrightarrow}N(0,\sigma_{G_1}^{2}),$$
where $\sigma_{G_1}^{2}=\mathbb{V}(G_1(Y_3))+2Cov(G_1(Y_3),G_1(Y_4))+2Cov(G_1(Y_3),G_1(Y_5)),$ due to $Y_3$ is indepedent of $Y_k$ for $k\geq 6$. Finally, making the required computations, the proof is completed.

\end{proof}
\begin{remark}\label{r1s2}
Conforming to the proof of Theorem \ref{t1s2}, the probability of transition in three steps just depends on the final state, and furthermore, the rows of $P^3$ are the initial distribution $\pi_3$. Using the result found for $P^3$, we can similarly express each entry of $P^4$ as follows
$$\mathbb{P}(Y_{7}=def \mid Y_3=abc)=\sum_{i\in\{1,2\}}p^{2-i}(1-p)^{i-1} \pi_3[def]=\pi_{3}[def].$$
Hence, $P^3=P^4$ and as a consequence $P^{n}=P^3,$ for $n\geq 3.$ It is noteworthy to emphasize that the aforementioned procedure provides an alternative method to demonstrate that the stationary distribution of the Markov Chain is  the initial distribution. By the way, the fact that the initial distribution is the stationary distribution indicates that the variables $Y_i$ are identically distributed, which means that they are independent of $n$. Specifically, $Y_{i} \sim  Mult(1, \hat{p})$ and $G_1(Y_{i}) \sim  \hat{g}_1 \cdot Mult(1, \hat{p})$  with $\hat{g}_1=(g_1(1,1,1),g_1(1,1,2),\dots,g_1(2,2,1),g_1(2,2,2))$ and $\hat{p}=(p^3,p^3(1-p),\dots,p(1-p)^2,(1-p)^3)$.
\end{remark}
The proof of the following theorem uses a similar approach as the proof of the previous theorem.
\begin{theorem}\label{t2s2}
As $n \rightarrow \infty$,
\hfill
\medskip

\begin{center}
    $\frac{ZC2_{n} - (75 - 23 p + 4 p^2 - 13 p^3 + 5 p^4)n}{\sqrt{n}} \stackrel{D}{\longrightarrow}N(0,\sigma^{2})$,
\end{center}
where $\sigma^2=11250-6796 p+2176 p^2-3690 p^3+1820 p^4-152 p^5.$
\end{theorem}

To complement the results found in Theorems \ref{t1s2} and \ref{t2s2}, let us consider the following propositions.
 
\begin{proposition}\label{p1}
Given a random polyomino chain with $n$ squares, it follows that
$$
    \mathbb{E}(ZC1_{n})=50n-92-6(3n-8)p+2(n-4)p^2(1-p),
$$
for $n\geq 4$ and
$$
\mathbb{V}(ZC1_{n})=4(-6972+(324n+2680)p+(-464n+1360)p^2+(364n-1584)p^3$$
$$+(-252n+1128)p^4+(48n-256)p^5+(-20n+104)p^6),
$$
for $n\geq 6$.   Moreover, $ZC1_{n}$ is asymptotically concentrated around its mean. More precisely, as $n\to \infty$, 
$\frac{ZC1_n}{n} \stackrel{a.s}{\longrightarrow} 50-18p+2p^2-2p^3.$

\end{proposition}

\begin{proof}

\noindent
  Note that the expression for the expectation follows directly from the recursive relation studied in the preceding section. With this initial result in hand, we can derive the expression for the variance by calculating the second moment and applying the recursive relation as follows
  $$
ZC1_n^2=ZC1_{n-1}^2+2ZC1_{n-1}g_1(L_{n-2},L_{n-1},L_{n})+g_1^2(L_{n-2},L_{n-1},L_{n}),$$ 
and
$$ZC1_{n-1}=ZC1_{n-3}+g_1(L_{n-4},L_{n-3},L_{n-2})+g_1(L_{n-3},L_{n-1},L_{n-1}),$$
with $ZC1_{n-3}$ independent of $g_1(L_{n-2},L_{n-1},L_{n})$. The second part is obtained by using the Strong Law of Large Numbers for Markov Chains \cite{bd1,bd2}.

\end{proof}

\begin{proposition}\label{p2}
Given a random polyomino chain with $n$ squares, it follows that
$$
    \mathbb{E}(ZC2_{n})=75n-159+(-23n+71)p+(4n-30)p^2+(-134n+65)p^3+5(n-5)p^4,
$$
for $n\geq 5$ and
$$
 \mathbb{V}(ZC2_{n})=33750n-115200+(-19893n+91222)p+(5279n-93671)p^2$$
$$~~~~~~~~~~~~~~~~~~~~~~~~+(-8799n+-127136)p^3
+(2593n-48519)p^4
+(1930n-11394)p^5$$
$$~~~~~~~~~~~~~~~~~+(-1751n+12579)p^6+(890n-5954)p^7+(-175n+1177)p^8,$$
for $n\geq 8$. Moreover, $ZC2_{n}$ is asymptotically concentrated around its mean. More precisely, as $n\to \infty$, 
$\frac{ZC2_n}{n} \stackrel{a.s}{\longrightarrow} 75 - 23 p + 4 p^2 - 13 p^3 + 5 p^4$.

\end{proposition}
\begin{proof}

\noindent
The demonstration proceeds similarly, with the only variation being that in this particular case,
 
$$ZC2_{n-1}=ZC1_{n-4}+g_2(L_{n-6},L_{n-5},L_{n-4},L_{n-3})+g_2(L_{n-5},L_{n-4},L_{n-3},L_{n-2})+g_2(L_{n-4},L_{n-3},L_{n-2},L_{n-1})$$
with $ZC2_{n-4}$ independent of $g_2(L_{n-3},L_{n-2},L_{n-1},L_n)$. The second part is performed as before.
 
\end{proof}

\begin{remark}
  As mentioned above, a recursive formula for degree-based topological indices over $PC_n$ was derived in \cite{bd14} and the asymptotic behavior and other characteristics for certain types of degree-based topological indices over $RPC_n$, such as the generalized Zagreb index, were also studied. Again, we have a recursive formula; only in this case, it is governed by a function of two variables. Thus, the same procedure (with slight differences) performed for $ZC1_n$ and $ZC2_n$ works, giving us information about the 
expected value, variance and asymptotic behavior of any degree-based topological index over $RPC_n$ (obviously encompassing those already covered in \cite{bd14}).
\end{remark}

\section{Conclusion}

\noindent

In this study, we have investigated the first and second Zagreb connection indices of polyomino chains and random polyomino chains using a Markov chain approach. Specifically, we have computed these indices within polyomino chains, explored extreme graphs, and outlined several patterns. In addition, we have formulated a central limit theorem concerning random polyomino chains. Finally, it would be interesting to extend the work of this paper to $k$-polyomino chains
or other types of topological indices following the same methodology.

\section*{Funding Information}

The authors were supported under Grant No. 00389.

\bibliographystyle{apsrev4-1}

\bibliography{Biblioteca}

\end{document}